\documentclass[journal,twoside,web,letter]{ieeecolor_noruler}
\usepackage{generic}
\usepackage{cite}
\usepackage{amsmath,amssymb,amsfonts}
\usepackage{graphicx}
\usepackage{xcolor}
\usepackage{algorithm,algorithmic}
\usepackage{hyperref}
\hypersetup{hidelinks=true}
\usepackage{textcomp}
\usepackage{wrapfig}
\usepackage{mathrsfs}

\newtheorem{corol}{Corollary}
\newtheorem{definition}{Definition}

\newtheorem{remark}{Remark}
\newtheorem{thm}{Theorem}

\newcommand{\im}{\operatorname{im}}
\newcommand{\Gr}{\operatorname{Gr}}
\newcommand{\Span}{\operatorname{span}}
\DeclareMathOperator{\rank}{rank}

\DeclareMathOperator{\row}{row}

\def\BibTeX{{\rm B\kern-.05em{\sc i\kern-.025em b}\kern-.08em
        T\kern-.1667em\lower.7ex\hbox{E}\kern-.125emX}}
\begin{document}
\def\ZZ{{\mathbb Z}}
\def\RR{{\mathbb R}}
\def\NN{{\mathbb N}}
\def\CC{{\mathbb C}}

\title{Complete Characterization of Minimum-Order Functional Observers from Darouach to Luenberger}
\author{Tyrone Fernando
\thanks{T.~Fernando is with the Department of Electrical, Electronic and Computer Engineering, University of Western Australia (UWA), 35 Stirling Highway, Crawley, WA 6009, Australia. (email: tyrone.fernando@uwa.edu.au)}}

\maketitle

\begin{abstract}
	This paper gives an order-by-order characterization of linear functional
	observers between the minimum dimension permitted by the algebraic observer
	condition and the upper endpoint determined by functional observability.
	Functional observability indices determine the minimum algebraically
	admissible order
	$$
	q_0=\sum_{j=1}^r\eta_j,
	$$
	where $r$ is the number of independent functionals to be estimated and
	$\eta_j$ is the functional observability index associated with the $j$th
	functional. For an $n$-dimensional system with $p$ independent measured outputs, let
	$n_0$ denote the rank of the observability matrix. Under functional
	observability, the upper endpoint of the functional-observer order range
	is $n_0-p$. The algebraic
	and spectral observer conditions are characterized at each prescribed
	order between these two endpoints. A nullspace representation in the
	original state space yields rank and matrix-pencil characterizations of
	the algebraic and spectral conditions, respectively. The minimum observer
	order is obtained as the smallest spectrally feasible order in the
	resulting order spectrum. The two endpoints include the classical
	Darouach and Luenberger orders as special cases: the lower endpoint
	reduces to the Darouach order when $q_0=r$, while, under complete
	observability, the upper endpoint reduces to the Luenberger order $n-p$.
\end{abstract}


\section{Introduction}
\label{sec:introduction}

Functional observers estimate prescribed linear functions of the state without
necessarily reconstructing the complete state. Their dimension may therefore
be substantially smaller than that required for state observation \cite{Luenberger1966}. The
determination of the smallest possible observer dimension has consequently
received sustained attention; see, among others,
\cite{FortmannWilliamson1972, Murdoch1973, RomanBullock1975, MooreLedwich1975, FurutaKawaji1977, Sirisena1979, darouach2000, ref9n, ref10, ref10n, RotellaZambettakis2011, HamdounEtAl2017}. More recently, functional-observer theory has been extended
to large-scale and interconnected systems~\cite{ref14nb, ref14nbb, ref14nd},
networked systems~\cite{ref14n, ref7n, ref5ty, mont}, sampled-data
systems~\cite{ref12n} and nonlinear systems \cite{ref13nb, nonl1, nonl3}.

Consider the linear system
$$
\dot x(t)=Ax(t)+Bu(t),\qquad y(t)=Cx(t),
$$
and the prescribed functional
$$
z(t)=Lx(t).
$$
Let $n$ denote the state dimension, $p=\rank C$, and let $r$ denote the
number of independent functional rows modulo the measured output.
The dimension $r$ is the natural lower bound on the observer
order and is the order of the classical Darouach
observer~\cite{darouach2000}.

In our previous work~\cite{fernando2026}, the functional observability 
indices were used to determine the minimum dimension permitted by the
algebraic observer condition. If $\eta_1,\ldots,\eta_r$ are the functional observability indices
associated with the $r$ functional rows, this dimension is
$$
q_0=\sum_{j=1}^r\eta_j\ge r.
$$
The same construction gives the complete family of order-$q_0$ augmented
functional matrices satisfying the algebraic condition, with observer
existence at this order determined by a spectral rank condition.

At the opposite end, let
$$
n_0:=\rank\left(\mathcal O_C\right),
$$
where $\mathcal O_C$ is the observability matrix of $(A,C)$. Under functional
observability, we show that the lower and upper endpoints of the admissible
functional-observer order range are $q_0$ and $n_0-p$, respectively. Hence, writing $\nu_{\min}$ for the minimum order of a
functional observer,
$$
r\le q_0\le\nu_{\min}\le n_0-p\le n-p.
$$
The classical Darouach and reduced-order Luenberger observer orders arise as
the limiting cases of these endpoints: $q_0=r$ gives the Darouach order,
whereas, when $(A,C)$ is observable, $n_0=n$ and the upper endpoint becomes
$n-p$, the order of the classical reduced-order Luenberger observer.

The question addressed in this paper is what lies between the algebraic
lower endpoint $q_0$ and the functional-observability upper endpoint
$n_0-p$. In particular, if every algebraically admissible realization at order $q_0$
fails the spectral condition, at which subsequent orders can a functional
observer exist, and can existence at each such order be characterized
without prescribing a particular augmentation mechanism?

The first main result concerns the algebraically admissible orders between
these two endpoints. Starting from the minimum algebraic order $q_0$, the
admissible orders are characterized up to the functional-observability
endpoint $n_0-p$. The second main result gives an exact matrix characterization of the
algebraic and spectral observer conditions at an arbitrary prescribed
order $q$. For each order $q$, a full-column-rank matrix $Z$ is used to
represent the nullspace associated with a reduced order-$q$ realization.
The analysis is carried out entirely in the original $n$-dimensional
state space. The algebraic condition becomes a rank test involving $Z$
and a basis of $\ker(CAZ)$, while the spectral condition becomes an
equivalent rank condition on a matrix pencil.

The principal contributions are therefore as follows. First, the complete
algebraically admissible family is characterized at every prescribed
admissible order, not only at the minimum algebraic order $q_0$. Second, the
algebraic and spectral observer conditions are converted into dual nullspace
tests expressed entirely in the original state coordinates. Third, a finite
matrix-chart parameterization covers all admissible nullspaces at a
prescribed order, so failure of the tests for every parameter value in
every chart, including the lower-dimensional strata on which the rank of
$CAZ$ drops, proves nonexistence at that order. Fourth, these fixed-order results are assembled into a complete
characterization of the feasible-order spectrum $\Sigma_{FO}(L)$ and
hence of the minimum functional-observer order.
Finally, the minimum-order characterization of our previous work~\cite{fernando2026} is recovered explicitly as the specialization $q=q_0$.

The resulting theory characterizes the functional-observer order spectrum
between $q_0$ and $n_0-p$, recovering the classical Darouach and
Luenberger orders as its two limiting cases. The minimum observer order
$\nu_{\min}$ is then obtained as the smallest order in this spectrum
satisfying both the algebraic and spectral observer conditions.

\section{Notation and Problem Formulation}
\label{sec:problem}

\subsection{Notation}

For a matrix $M$, $\row(M)$, $\rank(M)$, $\ker(M)$, $\im(M)$ and
$\operatorname{eig}(M)$ denote its row space, rank, nullspace, column
space, and spectrum (set of eigenvalues), respectively. The symbol $I$
denotes the identity matrix when its dimension is clear from context;
otherwise, $I_k$ denotes the $k\times k$ identity matrix. Matrices with no
columns are allowed: such a matrix has rank zero, its column space is
$\{0\}$, and the determinant of a $0\times0$ matrix is $1$. 

Let
$$
\mathcal O_C:=
\begin{pmatrix}
	C\\ CA\\ \vdots\\ CA^{n-1}
\end{pmatrix},
\qquad
\mathcal O_L:=
\begin{pmatrix}
	L\\ LA\\ \vdots\\ LA^{n-1}
\end{pmatrix}
$$
denote the observability matrices of $(A,C)$ and $(A,L)$, respectively.
We also use a full-column-rank matrix $V$ satisfying
\begin{equation}
	\im V=\ker\begin{pmatrix}C\\CA\end{pmatrix}.
	\label{eq:N1-basis}
\end{equation}
All nullspaces are taken over $\mathbb R$ unless complex values of
$\lambda$ are involved, in which case the corresponding complexified
spaces are understood.

\subsection{Problem Formulation}

Consider the linear time-invariant system
$$
\dot x(t)=Ax(t)+Bu(t),\qquad y(t)=Cx(t),
$$
where
$$
A\in\mathbb R^{n\times n},\quad
B\in\mathbb R^{n\times m},\quad
C\in\mathbb R^{p\times n}.
$$
The functional to be estimated is
$$
z(t)=Lx(t),\qquad L\in\mathbb R^{r\times n}.
$$
Throughout, redundant measured outputs and redundant functional rows are
removed so that
$$
\rank(C)=p,
\qquad
\rank\begin{pmatrix}C\\L\end{pmatrix}=p+r.
$$
Thus, components of $L$ already contained in $\row(C)$ are excluded from
the functional estimation problem.

Recall that
$$
n_0=\rank(\mathcal O_C)
$$
is the dimension of the observable part of $(A,C)$. Functional
observability is characterized by~\cite{ref3n}
$$
\rank
\begin{pmatrix}
	\mathcal O_C\\
	\mathcal O_L
\end{pmatrix}
=
\rank(\mathcal O_C)
=n_0.
$$
The original $n$-dimensional state-space representation is retained
throughout; no separate $n_0$-dimensional state realization is introduced.

Let $\mathcal L\in\mathbb R^{q\times n}$ be an augmented functional matrix
of order $q$. The prescribed functional $Lx(t)$ is required to be recoverable
from $\mathcal Lx(t)$ and the measured output $Cx(t)$. Equivalently,
$$
\row(L)
\subseteq
\row
\begin{pmatrix}
	C\\
	\mathcal L
\end{pmatrix}.
$$
In rank form, this condition is
$$
\rank
\begin{pmatrix}
	C\\
	\mathcal L\\
	L
\end{pmatrix}
=
\rank
\begin{pmatrix}
	C\\
	\mathcal L
\end{pmatrix}.
$$
A reduced order-$q$ realization further satisfies
$$
\rank
\begin{pmatrix}
	C\\
	\mathcal L
\end{pmatrix}
=
p+q.
$$
Thus, the observer order is determined directly by the rank increase
produced by adjoining $\mathcal L$ to $C$.
For estimating the augmented functional $\mathcal Lx(t)$, the algebraic
observer condition is
\begin{equation}
	\rank\begin{pmatrix}
		\mathcal LA\\
		\mathcal L\\
		CA\\
		C
	\end{pmatrix}
	=
	\rank\begin{pmatrix}
		\mathcal L\\
		CA\\
		C
	\end{pmatrix},
	\tag{A}\label{eq:A}
\end{equation}
and the spectral observer condition is
\begin{equation}
	\rank\begin{pmatrix}
		\lambda\mathcal L-\mathcal LA\\
		CA\\
		C
	\end{pmatrix}
	=
	\rank\begin{pmatrix}
		\mathcal L\\
		CA\\
		C
	\end{pmatrix},
	\quad \lambda\in\Omega_S,
	\tag{S}\label{eq:S}
\end{equation}
where $\Omega_S=\mathbb C$ for pole assignability. For the detectability
formulation of Darouach~\cite{darouach2000},
$$
\Omega_S=
\{\lambda\in\mathbb C:\Re(\lambda)\ge0\}.
$$
All real subspaces are understood through their complexifications when
complex values of $\lambda$ are considered.
By \cite{darouach2000,ref10}, applied to the augmented functional
$\mathcal Lx(t)$, an observer of order $q$ for $\mathcal Lx(t)$, whose
eigenvalues can be assigned arbitrarily, exists if and only if
conditions~(A) and~(S) hold with $\Omega_S=\mathbb C$; with the closed right
half-plane, (S) is the corresponding detectability condition. Throughout
this paper, functional observers, the order spectrum $\Sigma_{FO}(L)$ and
the minimum order $\nu_{\min}$ refer to arbitrary pole assignability,
$\Omega_S=\mathbb C$; the closed-right-half-plane version is used only
where detectability is mentioned explicitly.

The objective of this paper is to characterize the functional-observer
order spectrum within
$$
q_0\le q\le n_0-p,
$$
and, in particular, to determine the smallest order for which there exists
an augmented functional matrix $\mathcal L$ satisfying the recoverability
and reducedness requirements together with conditions~(A) and~(S).

\section{Algebraically Admissible Observer Orders}

\subsection{Order Endpoints}
\label{sec:FO-endpoint}

We recall only the part of the functional observability indices construction
required for the subsequent development. Write
$$
C=(C_1^T,\ldots,C_p^T)^T,
\qquad
L=(L_1^T,\ldots,L_r^T)^T,
$$
where $C_1,\ldots,C_p$ and $L_1,\ldots,L_r$ denote the rows of $C$ and
$L$, respectively. We refer to these rows as generators and order them as
\begin{equation}
	g_1,\ldots,g_{p+r}
	=
	C_1,\ldots,C_p,L_1,\ldots,L_r.
	\label{eq:generator-order}
\end{equation}
The candidates
$$
g_iA^k,\qquad k=0,1,2,\ldots,
$$
are processed by increasing powers of $A$ and, at each power, according to
the fixed generator order \eqref{eq:generator-order}. A candidate is appended
if and only if its inclusion increases the rank of the matrix constructed up
to that point. Since
$$
\rank\begin{pmatrix}C\\L\end{pmatrix}=p+r,
$$
all candidates corresponding to $k=0$ are appended. Hence the first $p+r$
rows of the functional observability matrix are
$$
\begin{pmatrix}C\\L\end{pmatrix}.
$$
For $k\ge1$, the candidates are therefore processed in the order
$$
C_1A^k,\ldots,C_pA^k,\;
L_1A^k,\ldots,L_rA^k.
$$
If $k_i$ denotes the largest power for which $g_iA^{k_i}$ is appended, define
$$
\alpha_i=1+k_i.
$$
The functional observability indices associated with the prescribed functional are
$$
\eta_j=\alpha_{p+j},
\qquad j=1,\ldots,r.
$$
Theorem~2 of~\cite{fernando2026} gives the minimum order for which
condition~(A) can be satisfied:
$$
q_0:=\sum_{j=1}^r\eta_j.
$$
Moreover, the complete family of order-$q_0$ augmented functional matrices
$\mathcal L$ satisfying condition~(A) is parameterized
in~\cite[Theorem~3]{fernando2026}. We denote this family by
$$
\mathfrak L_{q_0}^A(L)
:=
\left\{
\mathcal L\in\mathbb R^{q_0\times n}:
\begin{array}{l}
	\displaystyle
	\row\begin{pmatrix}C\\L\end{pmatrix}
	\subseteq
	\row\begin{pmatrix}C\\\mathcal L\end{pmatrix},\\[2mm]
	\displaystyle
	\rank\begin{pmatrix}C\\\mathcal L\end{pmatrix}=p+q_0,\\[2mm]
	\mathcal L\text{ satisfies condition~(A)}
\end{array}
\right\}.
$$
Thus $\mathfrak L_{q_0}^A(L)$ is the complete algebraically admissible
family at the minimum algebraic order, and no order $q<q_0$ can satisfy
condition~(A).

The lower endpoint established above, together with the upper endpoint
$n_0-p$, which is shown to be a functional-observer order in
Theorem~\ref{thm:FO-terminal}, therefore restricts the orders considered in
the sequel to
$$
q_0\le q\le n_0-p.
$$
The purpose of this section is to characterize the algebraically admissible
families throughout this entire interval.

\subsection{No-Gap Property}
\label{sec:no-A-gaps}

The rank formulation of condition~(A) permits an order-$q$ realization to be
extended by adjoining one independent row. The following result is the matrix
counterpart of the one-dimensional extension property.

\medskip

\begin{thm}[One-row admissible extension]
\label{thm:one-step}
Assume $(A,C,L)$ is functionally observable. Let
$\mathcal L\in\mathbb R^{q\times n}$ satisfy
\[
 \rank\begin{pmatrix}C\\\mathcal L\\L\end{pmatrix}
 =\rank\begin{pmatrix}C\\\mathcal L\end{pmatrix}=p+q,
\]
condition~\eqref{eq:A}, and
\begin{equation}
 \row\begin{pmatrix}C\\\mathcal L\end{pmatrix}\subseteq\row(\mathcal O_C).
\nonumber 
\end{equation}
If $q<n_0-p$, then there exists a row vector
$v\in\row(\mathcal O_C)$ such that
\begin{equation}
 \rank\begin{pmatrix}C\\\mathcal L\\v\end{pmatrix}=p+q+1,
\nonumber 
\end{equation}
and the augmented matrix
\[
 \mathcal L^+:=\begin{pmatrix}\mathcal L\\v\end{pmatrix}
\]
satisfies condition~\eqref{eq:A} at order $q+1$.
\end{thm}

\medskip

\begin{proof}
	Put
	\[
	X:=\begin{pmatrix}C\\\mathcal L\end{pmatrix},\qquad
	\mathcal H:=\begin{pmatrix}C\\CA\end{pmatrix}.
	\]
	Then $\rank X=p+q<n_0$ and
	$\row(X)\subseteq\row(\mathcal O_C)$. Condition~(A) is equivalent to
\[
\row(XA)\subseteq
\row\begin{pmatrix}X\\\mathcal H\end{pmatrix}.
\]
If $\row(\mathcal H)\subseteq\row(X)$, then
$$
\row(XA)\subseteq\row(X),
$$
so $\row(X)$ is right $A$-invariant. Since
$$
\row(C)\subseteq\row(X),
$$
it follows successively that
$$
\row(CA^k)\subseteq\row(X),
\qquad k=0,1,\ldots,n-1.
$$
Hence
$$
\row(\mathcal O_C)\subseteq\row(X),
$$
and therefore
$$
n_0=\rank(\mathcal O_C)\le\rank X.
$$
This contradicts $\rank X=p+q<n_0$.
Therefore
$$
\row(\mathcal H)\not\subseteq\row(X),
$$
and hence
$$
\rank\begin{pmatrix}X\\\mathcal H\end{pmatrix}>\rank X.
$$
By the Cayley--Hamilton theorem,
$\row(\mathcal O_C)A\subseteq\row(\mathcal O_C)$. Define
	\[
	\mathcal V:=\left\{v\in\row(\mathcal O_C):
	vA\in\row\begin{pmatrix}X\\\mathcal H\end{pmatrix}\right\}.
	\]
	This is a linear subspace of $\row(\mathcal O_C)$. To estimate its
	dimension, consider
	\[
	T:\row(\mathcal O_C)\longrightarrow\row(\mathcal O_C),
	\quad T(v)=vA,
	\]
	and let
	\[
	\pi:\row(\mathcal O_C)\longrightarrow
	\row(\mathcal O_C)/\row\begin{pmatrix}X\\\mathcal H\end{pmatrix}
	\]
	be the quotient map. By definition,
	$$
	\ker(\pi\circ T)=\mathcal V.
	$$
	Since $\dim\row(\mathcal O_C)=n_0$, rank--nullity gives
	$$
	\dim\mathcal V
	=
	n_0-\rank(\pi\circ T).
	$$
	Moreover,
	$$
	\rank(\pi\circ T)
	\le
	\dim\frac{\row(\mathcal O_C)}
	{\row\begin{pmatrix}X\\\mathcal H\end{pmatrix}}
	=
	n_0-\rank\begin{pmatrix}X\\\mathcal H\end{pmatrix}.
	$$
	Consequently,
	$$
	\dim\mathcal V
	\ge
	\rank\begin{pmatrix}X\\\mathcal H\end{pmatrix}
	>
	\rank X,
	$$
	where the strict inequality follows from
	$\row(\mathcal H)\nsubseteq\row(X)$.
	
	Since $\row(X)\subseteq\row(\mathcal O_C)$ and
	$$
	\row(XA)\subseteq
	\row\begin{pmatrix}X\\\mathcal H\end{pmatrix},
	$$
	the definition of $\mathcal V$ gives
	$$
	\row(X)\subseteq\mathcal V.
	$$
	Since
	$$
	\dim\mathcal V>\rank X=\dim\row(X),
	$$
	there exists
	$$
	v\in\mathcal V\setminus\row(X).
	$$
	Thus $v$ is independent of the rows of $X$, and hence
	\[
	\rank\begin{pmatrix}C\\\mathcal L\\v\end{pmatrix}
	=p+q+1.
	\]
	Moreover, by the definition of $\mathcal V$,
	$$
	vA\in\row\begin{pmatrix}X\\\mathcal H\end{pmatrix}.
	$$
	Together with
	$$
	\row(XA)\subseteq\row\begin{pmatrix}X\\\mathcal H\end{pmatrix},
	$$
	this yields
	$$
	\row\begin{pmatrix}XA\\vA\end{pmatrix}
	\subseteq
	\row\begin{pmatrix}X\\v\\\mathcal H\end{pmatrix}.
	$$
	Since
	$$
	\begin{pmatrix}X\\v\end{pmatrix}
	=
	\begin{pmatrix}C\\\mathcal L^+\end{pmatrix},
	$$
	this is precisely condition~(A) for $\mathcal L^+$.
\end{proof}

\medskip

The theorem is an extension result: starting from an admissible order-$q$
matrix, it constructs at least one admissible order-$(q+1)$ matrix by
adjoining one row, but it does not parameterize all admissible matrices at
the higher order. Combining it with the minimum condition-(A) order from
\cite{fernando2026} gives the following no-gap result.

\medskip

\begin{thm}[No gaps in condition-(A) feasibility]
\label{thm:A-spectrum}
Assume that the triple $(A,C,L)$ is functionally observable, and let
$q_0:=\sum_{j=1}^r\eta_j$. Then condition~(A) is feasible at every order
\[
 q\in\{q_0,q_0+1,\ldots,n_0-p\},
\]
whereas no order $q<q_0$ satisfies condition~(A).
\end{thm}

\medskip

\begin{proof}
	By \cite[Theorem~2]{fernando2026}, $q_0$ is the minimum order for which
	condition~(A) can be satisfied. Hence no order $q<q_0$ is algebraically
	admissible. The upper-bound part of that proof exhibits an order-$q_0$
	realization explicitly: by \cite[Theorem~1]{fernando2026}, the matrix
	$$
	\mathcal L
	=
	\begin{pmatrix}
		L\\
		R
	\end{pmatrix},
	\qquad
	R=\begin{pmatrix}R_1\\\vdots\\R_r\end{pmatrix},
	\qquad
	R_j=\begin{pmatrix}\ell_{j,\eta_j-1}\\\vdots\\\ell_{j,1}\end{pmatrix},
	$$
	satisfies condition~(A), where, by
	\cite[eq.~(10)]{fernando2026},
	$$
	\ell_{j,t}
	=
	L_jA^t
	-
	\sum_{q=\eta_j-t+1}^{\eta_j}
	\phi_{p+j}^{(q)}
	\mathcal G_{p+j}^{(q)}
	A^{q-\eta_j+t},
	$$
	for $t=1,\ldots,\eta_j-1$, where the coefficient rows
	$\phi_{p+j}^{(q)}$ and generator blocks $\mathcal G_{p+j}^{(q)}$ come from
	the first-dependence relation of $L_j$ \cite[Section~3]{fernando2026}.
	Each $\mathcal G_{p+j}^{(q)}$ consists of generator rows $g_i$ taken
	from $C$ and $L$, so every row of $\mathcal L$ is a linear combination
	of rows $C_iA^m$ and $L_iA^m$, $m\ge0$. The rows $C_iA^m$ lie in
	$\row(\mathcal O_C)$ by the Cayley--Hamilton theorem, and the rows
	$L_iA^m$ lie in $\row(\mathcal O_L)$. Since functional observability
	gives
	$$
	\row(\mathcal O_L)\subseteq\row(\mathcal O_C),
	$$
	it follows that
	$$
	\row
	\begin{pmatrix}
		C\\
		\mathcal L
	\end{pmatrix}
	\subseteq
	\row(\mathcal O_C).
	$$
	Moreover, $\rank\begin{pmatrix}C\\\mathcal L\end{pmatrix}=p+q_0$:
	otherwise some row of $\mathcal L$ would lie in the row space of $C$ and
	the remaining rows of $\mathcal L$, and deleting it would leave
	$\row\begin{pmatrix}C\\\mathcal L\end{pmatrix}$, and hence
	condition~(A), unchanged, contradicting the minimality of $q_0$.
	This is the required order-$q_0$ realization.
	
	Suppose now that a condition-(A) realization of order $q$ has been
	obtained with
	$$
	q_0\le q<n_0-p
	$$
	and
	$$
	\row
	\begin{pmatrix}
		C\\
		\mathcal L
	\end{pmatrix}
	\subseteq
	\row(\mathcal O_C).
	$$
	By Theorem~\ref{thm:one-step}, there exists
	$v\in\row(\mathcal O_C)$ such that
	$$
	\mathcal L^+
	=
	\begin{pmatrix}
		\mathcal L\\
		v
	\end{pmatrix}
	$$
	is a condition-(A) realization of order $q+1$. Moreover, because both
	$\row(C,\mathcal L)$ and $v$ belong to $\row(\mathcal O_C)$,
	$$
	\row
	\begin{pmatrix}
		C\\
		\mathcal L^+
	\end{pmatrix}
	\subseteq
	\row(\mathcal O_C).
	$$
	Thus the hypothesis required to apply Theorem~\ref{thm:one-step} is
	preserved after each extension.
	
	Starting from the order-$q_0$ realization and applying
	Theorem~\ref{thm:one-step} successively therefore yields condition-(A)
	realizations at every order
	$$
	q=q_0,q_0+1,\ldots,n_0-p.
	$$
	Together with the nonexistence of condition-(A) realizations for
	$q<q_0$, this proves the result.
\end{proof}

Accordingly, the set of orders at which condition~(A) is feasible is
\begin{equation}
	\{q_0,q_0+1,\ldots,n_0-p\}.
	\label{eq:SigmaA-result}
\end{equation}
When $(A,C)$ is observable, $n_0=n$, and \eqref{eq:SigmaA-result} becomes
$$\{q_0,q_0+1,\ldots,n-p\}.$$
The preceding theorem establishes feasibility, not a complete
parameterization. At $q=q_0$, the complete condition-(A) family
$\mathfrak L_{q_0}^{A}(L)$ is available from \cite{fernando2026}. Hence an order-$q_0$ functional observer exists if and only if at least
one member of $\mathfrak L_{q_0}^{A}(L)$ satisfies condition~(S), in which
case $\nu_{\min}=q_0$. Otherwise, the observer order must be increased.
For $q>q_0$, Theorem~\ref{thm:A-spectrum} guarantees condition-(A)
feasibility but does not characterize all admissible matrices at a
prescribed order. The next sections give the complete fixed-order
characterization using nullspace basis matrices.

\medskip

\begin{remark}
	The functional observability indices determine the lower endpoint $q_0$. Once $q_0$ is reached,
	condition~(A) is feasible at every successive order through $n_0-p$.
	The result does not imply corresponding monotonicity for the joint
	conditions~(A) and~(S).
\end{remark}

\section{Nullspace Matrix Characterization of Condition (A)} \label{sec:annihilator-A} We now seek a matrix characterization of all condition-(A) realizations at a prescribed order. Rather than parameterizing the rows of $\mathcal L$ directly, it is convenient to represent $\row\begin{pmatrix}C\\\mathcal L\end{pmatrix}$ through its nullspace. This leads to a column-space representation in which both the realization requirement and condition~(A) admit simple matrix forms.

\medskip

\begin{thm}[Matrix characterization of condition (A)]
	\label{thm:A-dual}
	Let $\mathcal L\in\mathbb R^{q\times n}$ be a reduced
	order-$q$ realization satisfying
	\[
	\rank\begin{pmatrix}C\\\mathcal L\end{pmatrix}=p+q,
	\quad
	\row\begin{pmatrix}
		C\\
		L
	\end{pmatrix}\subseteq\row\begin{pmatrix}
		C\\
		\mathcal L
	\end{pmatrix},
	\]
	and let $d=n-p-q$.
Let $Z\in\mathbb R^{n\times d}$ have full
	column rank and satisfy
	\[
	\im Z
	=
	\ker\begin{pmatrix}C\\\mathcal L\end{pmatrix}.
	\]
	Let $R_Z$ be any full-column-rank matrix satisfying
	\[
	\im R_Z=\ker(CAZ).
	\]
	Then condition~(A) is equivalent to
	\begin{equation}
		\im(AZR_Z)\subseteq\im Z.
	\nonumber 
	\end{equation}
	Equivalently,
	\begin{equation}
		\rank\begin{pmatrix}Z&AZR_Z\end{pmatrix}=d.
\nonumber 
	\end{equation}
\end{thm}

\begin{proof}
	Condition~(A) is equivalent to
	\[
	\row\begin{pmatrix}C\\\mathcal L\end{pmatrix}A
	\subseteq
	\row\begin{pmatrix}C\\\mathcal L\\CA\end{pmatrix}.
	\]
	Taking nullspaces gives the equivalent implication
	\[
	\begin{pmatrix}C\\\mathcal L\\CA\end{pmatrix}z=0
	\quad\Longrightarrow\quad
	\begin{pmatrix}C\\\mathcal L\end{pmatrix}Az=0.
	\]
	Since
	\[
	\im Z
	=
	\ker\begin{pmatrix}C\\\mathcal L\end{pmatrix}
	\]
	and $Z$ has full column rank, every vector satisfying
	$Cz=\mathcal Lz=0$ can be written uniquely as
	\[
	z=Z\xi.
	\]
	The additional condition $CAz=0$ is therefore equivalent to
	\[
	CAZ\xi=0,
	\]
	that is,
	\[
	\xi\in\ker(CAZ)=\im R_Z.
	\]
	Hence $\xi=R_Z\eta$ for some $\eta$, and consequently every vector
	satisfying the premise of the above implication has the form
	\[
	z=ZR_Z\eta.
	\]
	Substituting this expression into the conclusion of the above
	implication gives
	\[
	\begin{pmatrix}C\\\mathcal L\end{pmatrix}
	AZR_Z\eta=0
	\qquad\text{for every }\eta.
	\]
	Using again
	\[
	\ker\begin{pmatrix}C\\\mathcal L\end{pmatrix}
	=\im Z,
	\]
	this is equivalent to
	\[
	AZR_Z\eta\in\im Z
	\qquad\text{for every }\eta,
	\]
	or, equivalently,
	\[
	\im(AZR_Z)\subseteq\im Z.
	\]
	This proves the image inclusion.
	
	Finally, since $Z\in\mathbb R^{n\times d}$ has full column rank,
	$\rank Z=d$. Therefore,
	\[
	\im(AZR_Z)\subseteq\im Z
	\iff
	\rank\begin{pmatrix}Z&AZR_Z\end{pmatrix}=d,
	\]
	which proves the rank condition.
\end{proof}

\medskip

\begin{corol}[Fixed-order condition-(A) parameterization]
\label{cor:A-Z-rank}
For a prescribed $q\in\{q_0,\ldots,n_0-p\}$, put $d=n-p-q$.
The complete family of reduced matrices \(\mathcal L\in\mathbb R^{q\times n}\) satisfying the realization requirement and condition~(A) is parameterized by all full-column-rank matrices
$Z\in\mathbb R^{n\times d}$ satisfying
\[
 CZ=0,\quad LZ=0,
\]
and
\begin{equation}
 \rank\begin{pmatrix}Z&AZR_Z\end{pmatrix}=d,
 \quad \im R_Z=\ker(CAZ).
 \label{eq:cor-A-Z-rank}
\end{equation}
For each such $Z$, the corresponding matrices $\mathcal L$ are exactly
those satisfying
\[
	\ker\begin{pmatrix}C\\\mathcal L\end{pmatrix}=\im Z.
\]
\end{corol}

\medskip

\begin{remark}
	For a prescribed order $q$, Theorem~\ref{thm:A-dual} gives an
	equivalent nullspace condition for a given realization $\mathcal L$,
	whereas Corollary~\ref{cor:A-Z-rank} parameterizes all realizations
	$\mathcal L$ satisfying condition~(A) at that order.
\end{remark}

\medskip

Thus Corollary~\ref{cor:A-Z-rank} parameterizes the complete condition-(A)
family through nullspace basis matrices. The representation is not unique:
$Z$ and $ZQ$, with $Q$ nonsingular, define the same nullspace. Likewise,
two reduced matrices $\mathcal L$ and $\widehat{\mathcal L}$ represent the
same nullspace if and only if
\begin{equation}
 \widehat{\mathcal L}=Q\mathcal L+GC,
 \qquad Q\in\mathbb R^{q\times q}\ \text{nonsingular},\quad
 G\in\mathbb R^{q\times p}.
 \label{eq:equiv-Lcal}
\end{equation}
The matrix $R_Z$ is not an independent design variable; once $Z$ is fixed,
it is simply a basis matrix for $\ker(CAZ)$. The rank test
\eqref{eq:cor-A-Z-rank} states that multiplication by $A$ of the columns of
$ZR_Z$ produces no column direction outside $\im Z$.

\section{Nullspace Matrix Characterization of Condition (S)}
\label{sec:annihilator-S}

The spectral condition can be expressed using the same nullspace basis
matrix $Z$. Let $V$ be any full-column-rank matrix satisfying
\eqref{eq:N1-basis}.

\medskip

\begin{thm}[Matrix-pencil form of condition (S)]
	\label{thm:S-pencil}
	Assume condition~(A) holds, and let $Z$ be a full-column-rank matrix
	satisfying
	\[
	\im Z=\ker\begin{pmatrix}C\\\mathcal L\end{pmatrix}.
	\]
	Then condition~(S) is equivalent to
	\begin{equation}
		\rank\begin{pmatrix}(\lambda I-A)V&Z\end{pmatrix}
		=
		\rank\begin{pmatrix}V&Z\end{pmatrix},
		\quad \lambda\in\Omega_S.
		\label{eq:S-Z-rank}
	\end{equation}
\end{thm}

\medskip

\begin{proof}
	Let $v\in\im V$. Then $v=V\xi$ for some $\xi$. By
	\eqref{eq:N1-basis},
	$$
	Cv=CAv=0.
	$$
	Since
	$$
	\im Z
	=
	\ker\begin{pmatrix}C\\\mathcal L\end{pmatrix},
	$$
	we have
	\begin{equation}
		w\in\im Z
		\iff
		Cw=0,\quad \mathcal Lw=0,
		\qquad w\in\mathbb C^n.
		\label{eq:Z-membership}
	\end{equation}
	Taking $w=v$ and using $Cv=0$ gives
	$$
	v\in\im Z
	\iff
	\mathcal Lv=0.
	$$
	Taking $w=(\lambda I-A)v$ in \eqref{eq:Z-membership} gives
	$$
	(\lambda I-A)v\in\im Z
	$$
	if and only if
	$$
	C(\lambda I-A)v=0,
	\qquad
	\mathcal L(\lambda I-A)v=0.
	$$
	Since
	$$
	C(\lambda I-A)v
	=
	\lambda Cv-CAv
	=
	0,
	$$
	we obtain
	$$
	(\lambda I-A)v\in\im Z
	\iff
	(\lambda\mathcal L-\mathcal LA)v=0.
	$$
	
	Under condition~(A),
	$$
	\row(\mathcal LA)
	\subseteq
	\row\begin{pmatrix}
		\mathcal L\\
		CA\\
		C
	\end{pmatrix}.
	$$
	Since $\row(\mathcal L)$ is contained in the latter row space, it follows
	that
	$$
	\row\begin{pmatrix}
		\lambda\mathcal L-\mathcal LA\\
		CA\\
		C
	\end{pmatrix}
	\subseteq
	\row\begin{pmatrix}
		\mathcal L\\
		CA\\
		C
	\end{pmatrix}.
	$$
	Hence condition~(S) is equivalent to equality of these row spaces,
	and therefore to equality of their nullspaces. Thus
	$$
	\begin{pmatrix}
		\lambda\mathcal L-\mathcal LA\\CA\\C
	\end{pmatrix}v=0
	\iff
	\begin{pmatrix}
		\mathcal L\\CA\\C
	\end{pmatrix}v=0.
	$$
	Both nullspaces are contained in
	$$
	\ker\begin{pmatrix}C\\CA\end{pmatrix}
	=
	\im V.
	$$
	Hence it is sufficient to compare them on $\im V$.
	
	For $v\in\im V$, we have $Cv=CAv=0$. Hence the right-hand nullspace
	condition reduces to $\mathcal Lv=0$, which is equivalent to
	$v\in\im Z$, while the left-hand nullspace condition reduces to
	$(\lambda\mathcal L-\mathcal LA)v=0$, which is equivalent to
	$(\lambda I-A)v\in\im Z$. Therefore condition~(S) is equivalent to
	$$
	(\lambda I-A)v\in\im Z
	\iff
	v\in\im Z,
	\qquad v\in\im V.
	$$
	Writing $v=V\xi$, this becomes
	$$
	(\lambda I-A)V\xi\in\im Z
	\iff
	V\xi\in\im Z.
	$$
	
	Define
	\begin{IEEEeqnarray}{rcl}
		\mathscr X_\lambda
		&\ =\ &
		\{\xi:(\lambda I-A)V\xi\in\im Z\}
		\nonumber\\
		\mathscr X_*
		&\ =\ &
		\{\xi:V\xi\in\im Z\}.
		\nonumber
	\end{IEEEeqnarray}
	We first show that condition~(A) implies
	$$
	\mathscr X_*\subseteq\mathscr X_\lambda.
	$$
	Let $\xi\in\mathscr X_*$ and put $v=V\xi$. Then $v\in\im Z$, and hence
	$$
	\mathcal Lv=0.
	$$
	Since $v\in\im V$,
	$$
	Cv=CAv=0.
	$$
	Condition~(A) gives
	$$
	\row(\mathcal LA)
	\subseteq
	\row\begin{pmatrix}
		\mathcal L\\
		CA\\
		C
	\end{pmatrix},
	$$
	and therefore
	$$
	\mathcal LAv=0.
	$$
	Consequently,
	$$
	\mathcal L(\lambda I-A)v
	=
	\lambda\mathcal Lv-\mathcal LAv
	=
	0,
	$$
	and
	$$
	C(\lambda I-A)v
	=
	\lambda Cv-CAv
	=
	0.
	$$
	Hence
	$$
	(\lambda I-A)v
	\in
	\ker\begin{pmatrix}
		C\\
		\mathcal L
	\end{pmatrix}
	=
	\im Z,
	$$
	so $\xi\in\mathscr X_\lambda$. Therefore
	$$
	\mathscr X_*\subseteq\mathscr X_\lambda.
	$$
	
	Now
	$$
	(\lambda I-A)V\xi\in\im Z
	\iff
	V\xi\in\im Z
	$$
	for every $\xi$ if and only if
	$$
	\mathscr X_\lambda=\mathscr X_*.
	$$
	Since $\mathscr X_*\subseteq\mathscr X_\lambda$, this equality holds if
	and only if
	$$
	\dim\mathscr X_*=\dim\mathscr X_\lambda.
	$$
	By the rank--nullity theorem, this is equivalent to
	$$
	\rank\begin{pmatrix}
		(\lambda I-A)V&Z
	\end{pmatrix}
	=
	\rank\begin{pmatrix}
		V&Z
	\end{pmatrix},
	\qquad
	\lambda\in\Omega_S,
	$$
	which proves \eqref{eq:S-Z-rank}.
\end{proof}

\medskip

\begin{remark}[Basis invariance]
Replacing $Z$ by $ZQ$, where $Q$ is nonsingular, leaves
\eqref{eq:cor-A-Z-rank} and \eqref{eq:S-Z-rank} unchanged. Thus the tests
depend only on the column space of $Z$, although the formulation itself
uses matrices throughout.
\end{remark}

\section{Complete Fixed- and All-Order Characterization}
\subsection{Fixed-Order Characterization} \label{sec:fixed-order} At $q=q_0$, the complete condition-(A) family $\mathfrak L_{q_0}^{A}(L)$ is available from \cite{fernando2026}. For every prescribed order in the admissible range, Corollary~\ref{cor:A-Z-rank} parameterizes the complete condition-(A) family through $Z$. Applying the spectral condition of Theorem~\ref{thm:S-pencil} to this family yields the following necessary and sufficient fixed-order existence test. 

\medskip

\begin{thm}[Fixed-order functional-observer existence]
	\label{thm:main-fixed-order}
	Assume $(A,C,L)$ is functionally observable and
	\[
	q_0\le q\le n_0-p.
	\]
	Put
	\[
	d=n-p-q,
	\qquad
	r_1=\rank\begin{pmatrix}C\\CA\end{pmatrix},
	\]
	and let $V\in\mathbb R^{n\times(n-r_1)}$ have full column rank and satisfy
	\[
	\im V
	=
	\ker\begin{pmatrix}C\\CA\end{pmatrix}.
	\]
	Then an order-$q$ functional observer exists if and only if there is a
	full-column-rank matrix $Z\in\mathbb R^{n\times d}$ such that
	\begin{subequations}
		\label{eq:main-Z}
		\begin{align}
			CZ&=0,\quad LZ=0,\label{eq:main-Za}\\
			\rank\begin{pmatrix}Z&AZR_Z\end{pmatrix}&=d,
			\quad \im R_Z=\ker(CAZ),\label{eq:main-Zb}\\
			\rank\begin{pmatrix}(\lambda I-A)V&Z\end{pmatrix}
			&=\rank\begin{pmatrix}V&Z\end{pmatrix},
			\,\, \lambda\in\Omega_S.\label{eq:main-Zc}
		\end{align}
	\end{subequations}
\end{thm}

\bigskip

\begin{proof}
If an order-$q$ functional observer exists, choose a reduced augmented
functional matrix $\mathcal L\in\mathbb R^{q\times n}$ and a
full-column-rank matrix $Z\in\mathbb R^{n\times d}$ whose columns form
a basis of
$$
\ker\begin{pmatrix}
	C\\
	\mathcal L
\end{pmatrix}.
$$
	Since the prescribed functional $Lx$ is recoverable from $\mathcal Lx$
	and the measured output,
	$$
	\row(L)
	\subseteq
	\row\begin{pmatrix}
		C\\
		\mathcal L
	\end{pmatrix}.
	$$
	Hence
	$$
	CZ=0,
	\qquad
	LZ=0,
	$$
	which gives \eqref{eq:main-Za}. Theorem~\ref{thm:A-dual} gives
	\eqref{eq:main-Zb}, and Theorem~\ref{thm:S-pencil} gives
	\eqref{eq:main-Zc}.
	
	Conversely, suppose that a full-column-rank matrix $Z$ satisfies
	\eqref{eq:main-Z}. Since
	$$
	\rank Z=d=n-p-q
	$$
	and $CZ=0$, we have
	$$
	\im Z\subseteq\ker C,
	\qquad
	\dim(\im Z)=n-p-q.
	$$
	Therefore, there exists a matrix
	$\mathcal L\in\mathbb R^{q\times n}$ such that
	$$
	\ker\begin{pmatrix}
		C\\
		\mathcal L
	\end{pmatrix}
	=
	\im Z
	$$
	and
	$$
	\rank\begin{pmatrix}
		C\\
		\mathcal L
	\end{pmatrix}
	=
	p+q.
	$$
	Thus $\mathcal L$ is a reduced order-$q$ augmented functional matrix.
	
	Moreover, $LZ=0$ implies
	$$
	\im Z\subseteq\ker L.
	$$
	Since
	$$
	\im Z
	=
	\ker\begin{pmatrix}
		C\\
		\mathcal L
	\end{pmatrix},
	$$
	it follows that
	$$
	\ker\begin{pmatrix}
		C\\
		\mathcal L
	\end{pmatrix}
	\subseteq
	\ker L,
	$$
	or equivalently,
	$$
	\row(L)
	\subseteq
	\row\begin{pmatrix}
		C\\
		\mathcal L
	\end{pmatrix}.
	$$
	Hence the prescribed functional $Lx$ is recoverable from $\mathcal Lx$
	and the measured output.
	
	Finally, \eqref{eq:main-Zb}, together with
	Theorem~\ref{thm:A-dual}, gives condition~(A), while
	\eqref{eq:main-Zc}, together with Theorem~\ref{thm:S-pencil}, gives
	condition~(S). Hence an order-$q$ functional observer exists.
\end{proof}

\medskip

\begin{corol}[Recovery of the minimum-order family]
	\label{cor:q0-recovery}
At $q=q_0$, the nullspace characterization of
Theorem~\ref{thm:main-fixed-order} is equivalent to the complete minimum-order
characterization of \cite{fernando2026}. In particular, if the complete
condition-(A) family in \cite{fernando2026} is written as
\begin{equation}
 \mathcal L_J(\Theta,K)
 =
 \begin{pmatrix}
  L\\
  X_J(\Theta)M_C+KC
 \end{pmatrix},
 \label{eq:previous-min-family}
\end{equation}
where the index set $J$, the parameter matrix $\Theta$, and the matrices
$X_J(\Theta)$ and $M_C$ are as defined in \cite{fernando2026}, then, with $d_0=n-p-q_0$, its corresponding nullspace is represented by any
full-column-rank matrix $Z_J(\Theta)\in\mathbb R^{n\times d_0}$ satisfying
\begin{equation}
 \im Z_J(\Theta)
 =
 \ker\begin{pmatrix}
  C\\L\\X_J(\Theta)M_C
 \end{pmatrix}.
 \label{eq:Z-previous-family}
\end{equation}
The parameter $K$ does not affect this nullspace. Conversely, every
order-$q_0$ matrix $Z$ satisfying \eqref{eq:main-Za}--\eqref{eq:main-Zb}
corresponds, up to a change of nullspace basis and the row equivalence
\eqref{eq:equiv-Lcal}, to a member of \eqref{eq:previous-min-family}.
Moreover, \eqref{eq:main-Zc} is equivalent to condition~(S) of
\cite{fernando2026} for the corresponding $\mathcal L_J(\Theta,K)$.
\end{corol}

\medskip

\begin{proof}
	A vector $\xi\in\mathbb R^n$ belongs to
	$\ker\begin{pmatrix}C\\\mathcal L_J(\Theta,K)\end{pmatrix}$ if and only if
	\[
	C\xi=0,\qquad L\xi=0,\qquad
	\bigl(X_J(\Theta)M_C+KC\bigr)\xi=0.
	\]
	Since $C\xi=0$, the last equality is equivalent to
	$X_J(\Theta)M_C\xi=0$. Hence
	\[
	\ker\begin{pmatrix}C\\\mathcal L_J(\Theta,K)\end{pmatrix}
	=
	\ker\begin{pmatrix}C\\L\\X_J(\Theta)M_C\end{pmatrix},
	\]
which proves \eqref{eq:Z-previous-family} and shows explicitly why $K$
disappears from the nullspace representation. Since
$\mathcal L_J(\Theta,K)$ is reduced of order $q_0$, this nullspace has
dimension $d_0$.

Because every member of \eqref{eq:previous-min-family} satisfies
condition~(A), Theorem~\ref{thm:A-dual} gives
\begin{IEEEeqnarray}{rcl}
 \rank\begin{pmatrix}Z_J(\Theta)&
	AZ_J(\Theta)R_{Z_J(\Theta)}\end{pmatrix} &\ = \ & d_0 \nonumber \\
\im R_{Z_J(\Theta)} & = & \ker\bigl(CAZ_J(\Theta)\bigr). \nonumber 	
\end{IEEEeqnarray}	
Thus every member of the complete family of \cite{fernando2026} generates a member
of the present order-$q_0$ nullspace family.

Conversely, let $Z\in\mathbb R^{n\times d_0}$ satisfy
\eqref{eq:main-Za}--\eqref{eq:main-Zb}. By
Corollary~\ref{cor:A-Z-rank}, the reduced matrices $\mathcal L$ satisfying
\[
 \ker\begin{pmatrix}C\\\mathcal L\end{pmatrix}=\im Z
\]
are exactly the order-$q_0$ realizations satisfying condition~(A). Since the
parameterization \eqref{eq:previous-min-family} is complete, such an
$\mathcal L$ is represented, up to \eqref{eq:equiv-Lcal}, by
$\mathcal L_J(\Theta,K)$ for an admissible $(J,\Theta,K)$. Therefore
$\im Z=\im Z_J(\Theta)$, or equivalently
$Z=Z_J(\Theta)Q$ for some nonsingular $Q$. Finally,
Theorem~\ref{thm:S-pencil} shows that \eqref{eq:main-Zc} is equivalent to
the spectral condition~(S) for the corresponding $\mathcal L$. Hence the
present characterization recovers both the complete condition-(A) family and
the minimum-order existence test of \cite[Theorem~5]{fernando2026}.
\end{proof}

\medskip

\begin{remark}
Corollary~\ref{cor:q0-recovery} also exhibits a useful feature of the
nullspace formulation: the freedom $KC$ in the direct $\mathcal L$
parameterization is automatically quotiented out. The matrix $Z$ represents
the underlying subspace rather than a particular row realization of it.
\end{remark}

\subsection{Matrix Parameterization at a Prescribed Order}
\label{sec:grassmann}

Theorem~\ref{thm:main-fixed-order} characterizes fixed-order observer
existence in terms of a nullspace matrix $Z$. We now parameterize all
candidate matrices $Z$ at a prescribed order and then apply this
parameterization over the entire admissible order range.

Let $K_L\in\mathbb R^{n\times m_L}$ have full column rank and satisfy
\begin{equation}
	\im K_L
	=
	\ker\begin{pmatrix}C\\L\end{pmatrix}.
	\label{eq:KL-basis}
\end{equation}
Since $\rank\begin{pmatrix}C\\L\end{pmatrix}=p+r$, $m_L=n-p-r$.
For a prescribed $q\in\{q_0,\ldots,n_0-p\}$, put
\[
d=n-p-q.
\]
Since every candidate $Z\in\mathbb R^{n\times d}$ must satisfy
\[
CZ=0,\qquad LZ=0,
\]
its columns belong to $\im K_L$. Hence every full-column-rank candidate
$Z$ can be written as
\[
Z=K_LW,
\qquad
W\in\mathbb R^{m_L\times d},
\qquad
\rank W=d.
\]
The matrices $W$ and $WQ$, where $Q\in\mathbb R^{d\times d}$ is
nonsingular, have the same column space and therefore generate the same
candidate subspace through $Z=K_LW$. Thus only the column space of $W$
is relevant.

Since $\rank W=d$, there exist $d$ rows of $W$ that form a nonsingular
$d\times d$ submatrix. Let
\[
\mathcal I=\{i_1,\ldots,i_d\}
\subseteq
\{1,\ldots,m_L\}
\]
denote the indices of these $d$ rows. Let $P_{\mathcal I}$ be a
permutation matrix that reorders the rows of $W$ so that the $d$ rows
indexed by $\mathcal I$ occupy the first $d$ positions. Then write
\[
P_{\mathcal I}W
=
\begin{pmatrix}
	W_{\mathcal I}^{(1)}\\
	W_{\mathcal I}^{(2)}
\end{pmatrix},
\]
where $W_{\mathcal I}^{(1)}\in\mathbb R^{d\times d}$ is nonsingular.
Right multiplication by $\left(W_{\mathcal I}^{(1)}\right)^{-1}$ gives
\[
P_{\mathcal I}W\left(W_{\mathcal I}^{(1)}\right)^{-1}
=
\begin{pmatrix}
	I_d\\
	\Theta
\end{pmatrix},
\qquad
\Theta
=
W_{\mathcal I}^{(2)}
\left(W_{\mathcal I}^{(1)}\right)^{-1}.
\]
Hence
\begin{equation}
	W_{\mathcal I}(\Theta)
	:=
	P_{\mathcal I}^T
	\begin{pmatrix}
		I_d\\
		\Theta
	\end{pmatrix}
	=
	W\left(W_{\mathcal I}^{(1)}\right)^{-1},
	\qquad
	\Theta\in\mathbb R^{(m_L-d)\times d}.
	\label{eq:W-chart}
\end{equation}
Since $\left(W_{\mathcal I}^{(1)}\right)^{-1}$ is nonsingular,
\[
\im W_{\mathcal I}(\Theta)=\im W.
\]
The corresponding nullspace matrix is
\begin{equation}
	Z_{\mathcal I}(\Theta)
	=
	K_LW_{\mathcal I}(\Theta).
	\label{eq:Z-chart}
\end{equation}
By \eqref{eq:W-chart},
$$
Z_{\mathcal I}(\Theta)
=
Z\left(W_{\mathcal I}^{(1)}\right)^{-1}.
$$
Since $W_{\mathcal I}^{(1)}$ is nonsingular,
$$
\im Z_{\mathcal I}(\Theta)=\im Z.
$$
Each chart contains $d(m_L-d)$ free scalar parameters. The finitely many
choices of $\mathcal I$ cover all full-column-rank matrices $W$ up to
right multiplication by a nonsingular matrix. Equivalently, these are
the standard coordinate charts of the Grassmannian $\Gr(d,m_L)$, the set
of all $d$-dimensional subspaces of $\mathbb R^{m_L}$.

Geometrically, the fixed-order problem selects a $d$-dimensional subspace
\[
\mathcal Z\subseteq
\ker\begin{pmatrix}C\\L\end{pmatrix}
\]
that satisfies the $A$-compatibility and spectral requirements. The
matrices $W_{\mathcal I}(\Theta)$ provide finite coordinates for all such
subspaces; no geometric machinery is needed in the subsequent
computations.

The following result gives a complete chart-based parameterization of the
fixed-order characterization in Theorem~\ref{thm:main-fixed-order} and
reduces its spectral condition to a finite polynomial test.

\medskip

\begin{thm}[Matrix parameterization at a prescribed order]
 \label{thm:chart-fixed-order}
 For a prescribed $q\in\{q_0,\ldots,n_0-p\}$, put $d=n-p-q$.
 Let $K_L$ satisfy \eqref{eq:KL-basis}, and let $Z_\mathcal I(\Theta)$ be
 defined by \eqref{eq:W-chart}--\eqref{eq:Z-chart}. Then an
 order-$q$ functional observer exists if and only if, for at least
 one matrix chart, there exists
 $\Theta\in\mathbb R^{(m_L-d)\times d}$ such that
 \begin{IEEEeqnarray}{rcl}
  &\rank\begin{pmatrix}
   Z_\mathcal I(\Theta)&AZ_\mathcal I(\Theta)R_{Z_\mathcal I(\Theta)}
  \end{pmatrix}
  &=d,
  \IEEEyessubnumber \label{eq:chart-testa}\\
  &\rank\begin{pmatrix}
   (\lambda I-A)V&Z_\mathcal I(\Theta)
  \end{pmatrix}
  &=
  \rank\begin{pmatrix}
   V&Z_\mathcal I(\Theta)
  \end{pmatrix},\nonumber\\
  &&\hspace{-1.0cm}\lambda\in\Omega_S.
  \IEEEyessubnumber \label{eq:chart-testb}
 \end{IEEEeqnarray}
 where
 \[
  \im R_{Z_\mathcal I(\Theta)}
  =
  \ker\bigl(CAZ_\mathcal I(\Theta)\bigr).
 \]
 For a prescribed chart and parameter value $\Theta$, put
 \[
  h_{\mathcal I,\Theta}
  =
  \rank\begin{pmatrix}V&Z_\mathcal I(\Theta)\end{pmatrix},
 \]
and let $\Delta_j(\lambda)$ denote the nonzero
 $h_{\mathcal I,\Theta}\times h_{\mathcal I,\Theta}$ minors of
 \[
  P_{\mathcal I,\Theta}(\lambda)
  =
  \begin{pmatrix}
   (\lambda I-A)V&Z_\mathcal I(\Theta)
  \end{pmatrix}.
 \]
 Under \eqref{eq:chart-testa},
 \[
  \rank P_{\mathcal I,\Theta}(\lambda)\le h_{\mathcal I,\Theta},
  \qquad \lambda\in\mathbb C.
 \]
 Hence, with
 \[
  g_{\mathcal I,\Theta}(\lambda)
  =
  \gcd_\lambda\{\Delta_j(\lambda)\},
 \]
 condition \eqref{eq:chart-testb} is equivalent to
 \[
  g_{\mathcal I,\Theta}(\lambda)\ne0,
  \qquad \lambda\in\Omega_S.
 \]
 In particular, for pole assignability, $g_{\mathcal I,\Theta}$ is a nonzero
 constant and may be normalized to $g_{\mathcal I,\Theta}=1$. For
 detectability, $g_{\mathcal I,\Theta}$ has no zeros in the closed right
 half-plane.
\end{thm}

\medskip

\begin{proof}
	Since
	$$
	\im K_L=\ker\begin{pmatrix}C\\L\end{pmatrix},
	$$
	every full-column-rank matrix $Z\in\mathbb R^{n\times d}$ satisfying
	$CZ=LZ=0$ can be written as $Z=K_LW$, where
	$W\in\mathbb R^{m_L\times d}$ has rank $d$. Conversely, every such
	$W$ produces a full-column-rank matrix $Z=K_LW$ satisfying
	$CZ=LZ=0$.
	
	Since $\rank W=d$, some $d\times d$ row submatrix of $W$ is
	nonsingular. Let $\mathcal I$ denote its row-index set and let $P_\mathcal I$ be a
	permutation matrix that moves the rows indexed by $\mathcal I$ to the first $d$
	positions. Right multiplication by the inverse of the resulting leading
	$d\times d$ submatrix gives
	$$
	P_\mathcal IWQ=
	\begin{pmatrix}
		I_d\\
		\Theta
	\end{pmatrix}
	$$
	for some nonsingular $Q$. Equivalently,
	$$
	WQ
	=
	P_\mathcal I^T
	\begin{pmatrix}
		I_d\\
		\Theta
	\end{pmatrix}
	=
	W_\mathcal I(\Theta).
	$$
	Since $Q$ is nonsingular,
	$$
	\im(K_LWQ)=\im(K_LW).
	$$
	Hence $Z_\mathcal I(\Theta)=K_LW_\mathcal I(\Theta)$ represents the same nullspace as
	$Z=K_LW$. Varying $\mathcal I$ over all $d$-element subsets of
	$\{1,\ldots,m_L\}$ therefore covers all candidate nullspaces appearing
	in Theorem~\ref{thm:main-fixed-order}.
	
	The conditions $CZ=0$ and $LZ=0$ hold by construction. Substitution of
	$Z=Z_\mathcal I(\Theta)$ into the remaining conditions of
	Theorem~\ref{thm:main-fixed-order} therefore gives
	\eqref{eq:chart-testa}--\eqref{eq:chart-testb}.
	
	It remains to establish the finite spectral test. For brevity, write
	\begin{IEEEeqnarray}{rcl}
		Z&=&Z_\mathcal I(\Theta), \nonumber\\
		P(\lambda)&=&\begin{pmatrix}(\lambda I-A)V&Z\end{pmatrix}, \nonumber\\
		h&=&\rank\begin{pmatrix}V&Z\end{pmatrix}. \nonumber
	\end{IEEEeqnarray}
	Define
	\begin{IEEEeqnarray}{rcl}
		\mathscr X_\lambda
		&=&
		\{\xi:(\lambda I-A)V\xi\in\im Z\}, \nonumber\\
		\mathscr X_*
		&=&
		\{\xi:V\xi\in\im Z\}. \nonumber
	\end{IEEEeqnarray}
	By \eqref{eq:chart-testa}, condition~(A) holds. Hence, as in the
	proof of Theorem~\ref{thm:S-pencil},
	$$
	\mathscr X_*\subseteq\mathscr X_\lambda,
	\qquad
	\lambda\in\mathbb C.
	$$
	Let $r_V$ denote the number of columns of $V$. Since $Z$ has full
	column rank $d$, rank--nullity gives
	$$
	\dim\mathscr X_\lambda
	=
	r_V+d-\rank P(\lambda),
	$$
	and
	$$
	\dim\mathscr X_*
	=
	r_V+d-h.
	$$
	Therefore
	$\mathscr X_*\subseteq\mathscr X_\lambda$ implies
	$$
	\rank P(\lambda)\le h,
	\qquad
	\lambda\in\mathbb C.
	$$
	Consequently, $\rank P(\lambda)=h$ if and only if at least one
	$h\times h$ minor of $P(\lambda)$ is nonzero. Thus
	\eqref{eq:chart-testb} holds if and only if these minors have no
	common zero in $\Omega_S$. Their common zeros are precisely the zeros of
	$g_{\mathcal I,\Theta}(\lambda)$. Hence \eqref{eq:chart-testb} is equivalent to
	$$
	g_{\mathcal I,\Theta}(\lambda)\ne0,
	\qquad
	\lambda\in\Omega_S.
	$$
	For pole assignability, $\Omega_S=\mathbb C$, so the gcd is a nonzero
	constant and may be normalized to one. For detectability,
	$$
	\Omega_S=
	\{\lambda\in\mathbb C:\Re(\lambda)\ge0\},
	$$
	so the gcd must have no zeros in the closed right half-plane.
\end{proof}

\medskip

\begin{remark}[Rank strata of $CAZ_\mathcal I(\Theta)$]
	\label{rem:strata}
	The number of columns of $R_{Z_\mathcal I(\Theta)}$ is
	$\dim\ker\bigl(CAZ_\mathcal I(\Theta)\bigr)
	=d-\rank\bigl(CAZ_\mathcal I(\Theta)\bigr)$, which varies with
	$\Theta$. Condition~\eqref{eq:chart-testa} must therefore be imposed
	separately on each rank stratum
	$\{\Theta:\rank(CAZ_\mathcal I(\Theta))=\rho\}$,
	$\rho=0,1,\ldots,\min(p,d)$, including the lower-dimensional strata on
	which $CAZ_\mathcal I(\Theta)$ loses rank. An admissible nullspace may lie
	entirely on such a stratum (see Section~\ref{sec:example}). On a stratum,
	and on the subset where a chosen $\rho\times\rho$ minor of
	$CAZ_\mathcal I(\Theta)$ is nonzero, $R_{Z_\mathcal I(\Theta)}$ can be
	taken with entries rational in $\Theta$; clearing denominators, the rank
	conditions become polynomial equations together with the inequation that
	this minor is nonzero, and for pole assignability the requirement that the
	maximal minors in Theorem~\ref{thm:chart-fixed-order} have no common
	root can likewise be written polynomially (for example, through
	resultants). Theorem~\ref{thm:chart-fixed-order} characterizes existence:
	deciding it on each stratum of each chart is a real-algebraic feasibility
	problem over the finitely many real parameters in $\Theta$, decidable in
	principle by quantifier elimination \cite{BasuPollackRoy2006}.
	Algorithm~\ref{alg:all-order} is accordingly an order-by-order
	characterization rather than a specified numerical procedure; an
	exhaustive existence test must cover every stratum of every chart.
\end{remark}

\begin{remark}[Testing the fixed-order conditions]
	For a prescribed chart and parameter value $\Theta$,
	\eqref{eq:chart-testa} is tested by computing a basis
	$R_{Z_\mathcal I(\Theta)}$ of $\ker(CAZ_\mathcal I(\Theta))$ and evaluating the
	indicated rank. For pole assignability,
	\eqref{eq:chart-testb} is tested by the maximal-minor gcd condition
	in Theorem~\ref{thm:chart-fixed-order}, avoiding a pointwise search
	over $\lambda\in\mathbb C$.
\end{remark}

\begin{remark}[Observer construction]
	In general, once a reduced matrix $\mathcal L$ satisfying condition~(A)
	and condition~(S) at a prescribed order $q$ has been found -- whether
	via the chart parameterization above or by any other means -- the
	observer parameters can be constructed directly using the procedure
	in~\cite{darouach2000}, with the functional matrix $L$ therein replaced by this
	$\mathcal L$.
\end{remark}

\medskip

Thus the complete fixed-order search is reduced to finitely many matrix
charts, each parameterized by $d(m_L-d)$ scalar variables, stratified by
the rank of $CAZ_\mathcal I(\Theta)$ as in Remark~\ref{rem:strata}, and
subject to the two rank tests \eqref{eq:chart-testa}--\eqref{eq:chart-testb}. In
particular, the realization conditions $CZ=0$ and $LZ=0$ are satisfied
by construction.

\subsection{All-Order Characterization and Procedure}

We now pass from the fixed-order characterization to the complete
characterization of feasible observer orders.

\medskip

\begin{definition}[Observer families and order spectra]
	For each $q_0\le q\le n_0-p$, define
	\begin{IEEEeqnarray}{rcl}
		\mathfrak L_q^A(L)
		&:=&
		\{\mathcal L:\mathcal L\text{ is an order-$q$ realization satisfying (A)}\}
		\nonumber\\
		\mathfrak L_q^{FO}(L)
		&:=&
		\{\mathcal L\in\mathfrak L_q^A(L):
		\mathcal L\text{ satisfies (S)}\}.
		\nonumber
	\end{IEEEeqnarray}
	
	The corresponding order spectra are
	\[
	\Sigma_A(L):=\{q:\mathfrak L_q^A(L)\neq\varnothing\},
	\]
	\[
	\Sigma_{FO}(L):=\{q:\mathfrak L_q^{FO}(L)\neq\varnothing\}.
	\]
\end{definition}

	\medskip
	
By Theorem~\ref{thm:A-spectrum},
\[
\Sigma_A(L)=\{q_0,q_0+1,\ldots,n_0-p\}.
\]
Membership in $\Sigma_{FO}(L)$ is characterized at each prescribed order by
Theorem~\ref{thm:main-fixed-order}, with the complete matrix
parameterization and finite spectral test provided by
Theorem~\ref{thm:chart-fixed-order}. No interval property of
$\Sigma_{FO}(L)$ is assumed.

Parameterize the admissible orders by
\[
k_{\max}=n_0-p-q_0,
\qquad
q_k=q_0+k,
\qquad
d_k=n-p-q_k,
\]
for $k=0,1,\ldots,k_{\max}$. At the terminal order
$q_{k_{\max}}=n_0-p$,
\[
d_{k_{\max}}=n-n_0,
\]
which is zero only under full observability.

For each $k=0,1,\ldots,k_{\max}$, define the condition-$(A)$
nullspace family
$$
\mathfrak Z_{d_k}^{A}(L)
:=
\left\{
Z\in\mathbb R^{n\times d_k}:
\rank Z=d_k,\;
\eqref{eq:main-Za}\text{--}\eqref{eq:main-Zb}\text{ hold}
\right\},
$$
and the algebraic-spectral nullspace family
$$
\mathfrak Z_{d_k}^{AS}(L)
:=
\left\{
Z\in\mathfrak Z_{d_k}^{A}(L):
\eqref{eq:main-Zc}\text{ holds}
\right\}.
$$

Thus
$$
\mathfrak Z_{d_k}^{AS}(L)
\subseteq
\mathfrak Z_{d_k}^{A}(L).
$$

Since these conditions depend only on $\im Z$, matrices related by
$$
Z\sim ZQ,
\qquad
Q\in\mathbb R^{d_k\times d_k}
\quad\text{nonsingular},
$$
represent the same candidate nullspace.

\medskip

\subsubsection{Terminal Functional-Observability Endpoint}
\label{sec:Luenberger}

The upper endpoint $n_0-p$ is spectrally feasible. In matrix terms, its
nullspace basis is any full-column-rank matrix $Z_0$ satisfying
\[
\im Z_0=\ker(\mathcal O_C),
\qquad
\rank Z_0=n-n_0.
\]

\medskip
\begin{thm}[Functional-observability endpoint]
	\label{thm:FO-terminal}
	Assume $(A,C,L)$ is functionally observable. Then
	\[
	n_0-p\in\Sigma_{FO}(L).
	\]
\end{thm}

\medskip

\begin{proof}
	Functional observability gives
	\[
	LZ_0=0,
	\]
	while
	\[
	CZ_0=0
	\]
	follows from
	\[
	\im Z_0=\ker(\mathcal O_C).
	\]
	Since $\ker(\mathcal O_C)$ is $A$-invariant, there exists a matrix
	$F_0$ such that
	\[
	AZ_0=Z_0F_0.
	\]
	Hence \eqref{eq:main-Zb} holds automatically.
	
	To verify \eqref{eq:main-Zc}, let $v=V\xi$ and suppose
	\[
	(\lambda I-A)v\in\im Z_0.
	\]
	Since
	\[
	\im Z_0=\ker(\mathcal O_C)
	\]
	and $\ker(\mathcal O_C)$ is $A$-invariant, we have
	\[
	CA^k(\lambda I-A)v=0,
	\qquad k\ge0.
	\]
	Therefore
	\[
	\lambda CA^kv-CA^{k+1}v=0,
	\qquad k\ge0.
	\]
	Moreover, since $v\in\im V$ and
	\[
	\im V=
	\ker
	\begin{pmatrix}
		C\\CA
	\end{pmatrix},
	\]
	we have
	\[
	Cv=0,\qquad CAv=0.
	\]
	The preceding recurrence then gives successively
	\[
	CA^kv=0,
	\qquad k\ge0.
	\]
	Hence
	\[
	v\in\ker(\mathcal O_C)=\im Z_0.
	\]
	
	Conversely, if $v\in\im Z_0$, then the $A$-invariance of
	$\im Z_0$ gives
	\[
	(\lambda I-A)v\in\im Z_0
	\]
	for every $\lambda\in\mathbb C$. Therefore
	\[
	(\lambda I-A)v\in\im Z_0
	\quad\Longleftrightarrow\quad
	v\in\im Z_0,
	\qquad v\in\im V.
	\]
	Thus the two solution sets underlying \eqref{eq:S-Z-rank} coincide for
	every $\lambda\in\Omega_S$, and condition~(S) holds. Therefore
	$n_0-p\in\Sigma_{FO}(L)$.
\end{proof}

\medskip

\begin{corol}[Luenberger endpoint]
	\label{cor:Luenberger}
	If $(A,C)$ is observable, then $n_0=n$, $Z_0$ has no columns, and
	\[
	n-p\in\Sigma_{FO}(L).
	\]
	Thus the functional-observability endpoint reduces to the classical
	reduced-order Luenberger order.
\end{corol}

\medskip

Then Theorem~\ref{thm:main-fixed-order} gives
\[
q_k\in\Sigma_{FO}(L)
\quad\Longleftrightarrow\quad
\mathfrak Z_{d_k}^{AS}(L)\ne\varnothing.
\]
Moreover, Theorem~\ref{thm:chart-fixed-order} provides a complete matrix
parameterization for testing the nonemptiness of
$\mathfrak Z_{d_k}^{AS}(L)$ at each prescribed order.
We first characterize functional observability itself before determining
the complete spectrum of feasible observer orders.

\medskip

\begin{thm}[Functional-observability characterization]
	\label{thm:FO-equivalence}
	The triple $(A,C,L)$ is functionally observable if and only if
	\[
	\Sigma_{FO}(L)\neq\varnothing.
	\]
	Equivalently, there exist
	$q\in\{q_0,\ldots,n_0-p\}$ and a full-column-rank matrix
	$Z\in\mathbb R^{n\times(n-p-q)}$ satisfying
	\eqref{eq:main-Za}--\eqref{eq:main-Zc}.
\end{thm}

\medskip

\begin{proof}
	If $(A,C,L)$ is functionally observable, then
	Theorem~\ref{thm:FO-terminal} guarantees a functional observer at
	$q=n_0-p$. Hence
	$\Sigma_{FO}(L)\neq\varnothing$.
	
	Conversely, if $\Sigma_{FO}(L)\neq\varnothing$, then a functional
	observer exists at some admissible order. Its augmented functional
	matrix satisfies conditions~(A) and~(S), and hence $(A,C,L)$ is
	functionally observable by the augmentation characterization of
	functional observability \cite{ref10}.
	
	The equivalent characterization in terms of $Z$ follows directly from
	Theorem~\ref{thm:main-fixed-order}.
\end{proof}

%
%

\medskip

Having established when a functional observer exists, we now characterize
the complete spectrum of feasible observer orders and determine its minimum.

\medskip

\begin{thm}[All-Order Observer Characterization]
\label{thm:all-order}
Assume $(A,C,L)$ is functionally observable. Then the complete feasible-order
spectrum is
\begin{equation}
 \Sigma_{FO}(L)
 =
 \{q_k:\mathfrak Z_{d_k}^{AS}(L)\ne\varnothing,\;
 k=0,1,\ldots,k_{\max}\}.
 \label{eq:complete-spectrum}
\end{equation}
Equivalently, $q_k$ is a functional-observer order if and only if at least
one chart $Z_\mathcal I(\Theta)$ at dimension $d_k$ satisfies
\eqref{eq:chart-testa}--\eqref{eq:chart-testb}. Moreover,
\begin{equation}
 \nu_{\min}
 =q_0+\min\{k:\mathfrak Z_{d_k}^{AS}(L)\ne\varnothing\}.
 \label{eq:min-k}
\end{equation}
\end{thm}

\medskip

\begin{proof}
For each $k$, Theorem~\ref{thm:main-fixed-order} gives
$q_k\in\Sigma_{FO}(L)$ if and only if
$\mathfrak Z_{d_k}^{AS}(L)\ne\varnothing$. Theorem~\ref{thm:chart-fixed-order}
covers every full-rank candidate $Z$ through finitely many matrix charts, so
this test is complete at each order. Taking the union over
$k=0,\ldots,k_{\max}$ proves \eqref{eq:complete-spectrum}. By
Theorem~\ref{thm:FO-terminal}, the terminal set is nonempty, so the minimum
in \eqref{eq:min-k} exists and gives the smallest feasible order.
\end{proof}

\medskip

Theorem~\ref{thm:all-order} is a completeness statement: at a prescribed
order, failure of \eqref{eq:chart-testa}--\eqref{eq:chart-testb} for every
$\Theta$ in every chart, on every rank stratum of Remark~\ref{rem:strata},
proves nonexistence of a functional observer at that order; it is not
merely failure of a selected realization.

\begin{algorithm}[t]
	\caption{Order-by-order functional-observer characterization}
	\label{alg:all-order}
	\begin{algorithmic}[1]
		\STATE Compute the FO indices $\eta_j$ and
		$q_0=\sum_j\eta_j$.
		\STATE Compute $n_0=\rank\mathcal O_C$ and
		$k_{\max}=n_0-p-q_0$.
		\STATE Compute basis matrices $V$ and $K_L$ from \eqref{eq:N1-basis} and 
		\eqref{eq:KL-basis} respectively.
		\FOR{$k=0,1,\ldots,k_{\max}$}
		\STATE Set $q_k=q_0+k$ and $d_k=n-p-q_k$.
		\IF{$k=0$}
		\STATE Use $\mathfrak L_{q_0}^{A}(L)$ from
		\cite{fernando2026} and test the spectral condition.
		\ELSIF{$k=k_{\max}$}
		\STATE Use a basis $Z_0$ of $\ker(\mathcal O_C)$
		and test the spectral condition.
		\ELSE
		\STATE Parameterize
		$Z=K_LW_\mathcal I(\Theta)$ using
		\eqref{eq:W-chart}--\eqref{eq:Z-chart}.
		\STATE Test
		\eqref{eq:chart-testa}--\eqref{eq:chart-testb} on every rank
		stratum of every chart (Remark~\ref{rem:strata}).
		\ENDIF
		\STATE Record whether $q_k\in\Sigma_{FO}(L)$.
		\ENDFOR
		\STATE Return $\Sigma_{FO}(L)$ and
		$\nu_{\min}=\min\Sigma_{FO}(L)$.
	\end{algorithmic}
\end{algorithm}

\section{Illustrative Example}
\label{sec:example}

We illustrate the all-order characterization with an example in which
the algebraic lower bound $q_0$ is not a functional-observer order.
Condition~$(A)$ first becomes possible at $q=q_0$, but condition~$(S)$
excludes every admissible realization at $q=q_0$. A functional observer
first exists at $q=q_0+1$, strictly below the functional-observability
endpoint $n_0-p$. The example also shows that an admissible nullspace may
lie on a lower-dimensional rank stratum of a chart
(Remark~\ref{rem:strata}).

Consider
$$
\dot x(t)=Ax(t)+Bu(t),\quad y(t)=Cx(t),\quad z(t)=Lx(t),
$$
where
\begin{IEEEeqnarray}{rcl}
A &\ = &\ 
\begin{pmatrix}
	0&1&0&1&0&0&0&0\\
	0&-1&1&0&0&0&0&-1\\
	1&0&-1&1&0&0&1&-1\\
	-1&-1&0&-2&0&0&0&0\\
	0&0&0&0&-2&0&0&1\\
	0&0&0&0&0&-2&1&0\\
	0&0&0&0&0&-1&-2&1\\
	0&0&0&0&0&0&1&-2
\end{pmatrix} \vspace*{2mm}
\nonumber \\
B &\ = &\  e_1+e_8 \vspace*{2mm} \nonumber \\
C &\ = &\ 
\begin{pmatrix}
	-2&-2&2&-2&-1&0&1&-3\\
	0&0&1&2&2&3&-1&-2
\end{pmatrix} \vspace*{2mm} \nonumber \\
L&\ = &\ 
\begin{pmatrix}
	-2&-1&-2&-3&-1&3&2&0\\
	1&-1&-2&1&0&2&0&1
\end{pmatrix}\nonumber 
\end{IEEEeqnarray}
where \(e_i\) is the \(i\)-th standard basis vector of \(\mathbb R^8\).
Thus
$$
n=8,\qquad p=\rank C=2,\qquad r=\rank L=2.
$$
The observability matrix satisfies
$$
n_0=\rank\mathcal O_C=8=n,
$$
so that $(A,C)$ is completely observable.
Applying the functional-observability-index construction gives
$$
\eta_1=2,\qquad \eta_2=1,
\qquad
q_0=\eta_1+\eta_2=3>r=2.
$$
Since $(A,C)$ is observable, the functional-observability endpoint reduces
to the classical reduced-order Luenberger order $n_0-p=6$, and the complete
order range that must be considered is
$$
q\in\{3,4,5,6\}.
$$
For this example,
$$
\dim\ker\begin{pmatrix}C\\L\end{pmatrix}=4,
\,\,
\dim\ker\begin{pmatrix}C\\CA\end{pmatrix}=4,
\,\,
\dim\ker\begin{pmatrix}C\\L\\CA\end{pmatrix}=2,
$$
and
$$
\ker\begin{pmatrix}C\\L\\CA\\LA\end{pmatrix}=\Span\{x_0\},
$$
\begin{IEEEeqnarray}{rcl}
x_0&\ =\ &(-390,\,97,\,-22,\,272,\,-419,\,58,\,-252,\,55)^T.\nonumber
\end{IEEEeqnarray}
Since $Cx_0=Lx_0=CAx_0=LAx_0=0$, the vector
\begin{IEEEeqnarray}{rcl}
Ax_0&\ =\ &(369,\,-174,\,-403,\,-251,\nonumber\\
&&\;\;893,\,-368,\,501,\,-362)^T\nonumber
\end{IEEEeqnarray}
also belongs to $\ker\begin{pmatrix}C\\L\end{pmatrix}$, while
$CA^2x_0=\begin{pmatrix}215&-41\end{pmatrix}^T\neq0$. In particular,
$x_0$ and $Ax_0$ are linearly independent. Put
$$
\mathcal Z_4:=\Span\{x_0,\,Ax_0\}.
$$
All computations below are exact, in rational arithmetic.

\subsection{Failure at the lower bound $q=q_0=3$}

At $q=q_0=3$, $d=n-p-q=3$, so $\im Z$ is a three-dimensional subspace of
the four-dimensional space $\ker\begin{pmatrix}C\\L\end{pmatrix}$.

\emph{Condition~(A).} Since $\im Z\subseteq\ker C$,
$$
\im(ZR_Z)=\im Z\cap\ker(CA)=\im Z\cap\ker\begin{pmatrix}C\\L\\CA\end{pmatrix},
$$
which has dimension at least $3+2-4=1$. By Theorem~\ref{thm:A-dual},
condition~(A) requires
$A\bigl(\im Z\cap\ker(CA)\bigr)\subseteq\im Z\subseteq\ker L$, and hence
$$
\im Z\cap\ker(CA)\subseteq
\ker\begin{pmatrix}C\\L\\CA\\LA\end{pmatrix}=\Span\{x_0\}.
$$
Therefore condition~(A) forces $\im Z\cap\ker(CA)=\Span\{x_0\}$ and
$Ax_0\in\im Z$; that is, every order-$3$ candidate satisfying
condition~(A) contains $\mathcal Z_4$. Completing $\{x_0,Ax_0\}$ to a basis
of $\ker\begin{pmatrix}C\\L\end{pmatrix}$ with
\begin{IEEEeqnarray}{rcl}
k_1&\ =\ &(36,\,-11,\,6,\,-35,\,32,\,0,\,0,\,0)^T,\nonumber\\
k_2&\ =\ &(52,\,5,\,22,\,-35,\,0,\,16,\,0,\,0)^T,\nonumber
\end{IEEEeqnarray}
these candidates form the one-parameter pencil
\begin{IEEEeqnarray}{rcl}
Z(t)&\ =\ &\begin{pmatrix}x_0&Ax_0&k_1+tk_2\end{pmatrix},\quad t\in\mathbb R,\nonumber\\
Z(\infty)&\ =\ &\begin{pmatrix}x_0&Ax_0&k_2\end{pmatrix}.\nonumber
\end{IEEEeqnarray}
Along the pencil,
$$
\det\Bigl(CA\begin{pmatrix}Ax_0&k_1+tk_2\end{pmatrix}\Bigr)=-13751\,t-8343,
$$
which vanishes only at $t^*=-8343/13751$. At $t=t^*$,
$\dim\ker(CAZ(t^*))=2$ and condition~(A) fails. For every other member,
including $t=\infty$, $\ker(CAZ)$ is spanned by the coordinate vector of
$x_0$, and condition~(A) holds because $Ax_0\in\im Z$. Hence the complete
condition-(A) family at $q=3$ is the pencil with $t^*$ removed.

\emph{Condition~(S).} For every member of this family,
$\im Z\cap\im V=\im Z\cap\ker(CA)=\Span\{x_0\}$, so
$h=\rank\begin{pmatrix}V&Z\end{pmatrix}=4+3-1=6$. The greatest common
divisor in $\mathbb Q[t,\lambda]$ of the $6\times6$ minors of
$\begin{pmatrix}(\lambda I-A)V&Z(t)\end{pmatrix}$ is, up to a constant
factor,
\begin{IEEEeqnarray}{rcl}
g(\lambda,t)&\ =\ &(13751t+8343)\lambda^3+\,(63951t+59247)\lambda^2\nonumber\\
&&+\,(134797t+133197)\lambda+\,(96617t+96105).\nonumber
\end{IEEEeqnarray}
Because $g$ divides every minor as a polynomial in $(t,\lambda)$, at each
fixed $t\neq t^*$ the minors share the factor $g(\lambda,t)$, which has
degree three in $\lambda$. At $t=\infty$ the corresponding greatest common
divisor is $13751\lambda^3+63951\lambda^2+134797\lambda+96617$. Hence every
member of the condition-(A) family loses rank in
$\begin{pmatrix}(\lambda I-A)V&Z\end{pmatrix}$ at some $\lambda\in\mathbb C$,
and condition~(S) fails throughout. Therefore no order-$3$ functional
observer exists:
$$
3\notin\Sigma_{FO}(L).
$$
Thus $q_0=3$ is the algebraic lower bound predicted by the
functional-observability indices, but it is not itself a
functional-observer order. The unassignable eigenvalues include the roots
of $g(\lambda,t)$; they move with $t$ but are present for every member of
the family.

\subsection{Feasibility at $q=q_0+1=4$}

At $q=4$, $d=2$. Take
$$
Z_4=\begin{pmatrix}x_0&Ax_0\end{pmatrix},\qquad \im Z_4=\mathcal Z_4.
$$
Then $\rank Z_4=2$, $CZ_4=0$ and $LZ_4=0$. Moreover,
$$
CAZ_4=\begin{pmatrix}0&215\\0&-41\end{pmatrix},
$$
so $\ker(CAZ_4)$ is spanned by $R_{Z_4}=\begin{pmatrix}1&0\end{pmatrix}^T$
and $AZ_4R_{Z_4}=Ax_0\in\im Z_4$. Hence
$\rank\begin{pmatrix}Z_4&AZ_4R_{Z_4}\end{pmatrix}=2$, and condition~(A)
holds. For condition~(S),
$h=\rank\begin{pmatrix}V&Z_4\end{pmatrix}=5$, since $x_0\in\im V$, and the
greatest common divisor of the $5\times5$ minors of
$\begin{pmatrix}(\lambda I-A)V&Z_4\end{pmatrix}$ is
$$
\gcd_\lambda\{\Delta_j(\lambda)\}=1.
$$
Hence
$\rank\begin{pmatrix}(\lambda I-A)V&Z_4\end{pmatrix}
=\rank\begin{pmatrix}V&Z_4\end{pmatrix}=5$ for all $\lambda\in\mathbb C$,
and condition~(S) holds. Therefore
$$
4\in\Sigma_{FO}(L),
$$
and, since $q=3$ is infeasible,
$$
\nu_{\min}=4=q_0+1.
$$
Note that $\im Z_4$ lies on the lower-dimensional stratum
$\rank(CAZ)=1$ of the charts of $\Gr(2,4)$, the situation described in
Remark~\ref{rem:strata}.

\subsection{Feasibility at $q=q_0+2=5$}

At $q=5$, $d=1$. The one-dimensional subspace generated by
$$
Z_5=
\begin{pmatrix}
	1&1&1&0&-2&1&0&0
\end{pmatrix}^T
$$
satisfies $CZ_5=0$, $LZ_5=0$ and
$CAZ_5=\begin{pmatrix}-3&-1\end{pmatrix}^T$, which has full column rank.
Hence $R_{Z_5}$ has no columns and condition~(A) holds trivially. With
$h=\rank\begin{pmatrix}V&Z_5\end{pmatrix}=5$, the greatest common divisor
of the $5\times5$ minors of
$\begin{pmatrix}(\lambda I-A)V&Z_5\end{pmatrix}$ equals $1$, so
condition~(S) holds and
$$
5\in\Sigma_{FO}(L).
$$

\subsection{Existence at the functional-observability endpoint}

At $q=n_0-p=6$, $d=0$, the terminal nullspace is trivial, and
Theorem~\ref{thm:FO-terminal} gives
$$
6\in\Sigma_{FO}(L).
$$
Combining the four admissible orders gives
$$
\Sigma_{FO}(L)=\{4,5,6\},
\qquad
\nu_{\min}=4=q_0+1<n_0-p=6.
$$

\subsection{Explicit Observer Realization at $q=4$}

An explicit reduced matrix $\mathcal L$ with
$\ker\begin{pmatrix}C\\\mathcal L\end{pmatrix}=\im Z_4$ is obtained by
augmenting $L$ with two rows that annihilate $x_0$ and $Ax_0$:
$$
\mathcal L=\begin{pmatrix}L\\u_1\\u_2\end{pmatrix},
$$
where
\begin{IEEEeqnarray}{rcl}
u_1&\ =\ &\bigl(1,\,\tfrac{354}{3283},\,0,\,\tfrac{4581}{3283},\,0,\,0,\,0,\,0\bigr),\nonumber\\
u_2&\ =\ &\bigl(0,\,1,\,0,\,0,\,0,\,-\tfrac{6386}{189},\,0,\,\tfrac{6401}{189}\bigr).\nonumber
\end{IEEEeqnarray}
The first two rows of $\mathcal L$ are exactly $L$, so the first two
components of the estimate of $\mathcal Lx(t)$ estimate $z(t)=Lx(t)$; the
remaining two rows complete $\mathcal L$ to order $4$.

With $\mathcal L$ fixed, the observer
\begin{IEEEeqnarray}{rcl}
\dot w(t) \ &=& \ Nw(t)+Jy(t)+Hu(t) \nonumber \\
\widehat{\mathcal Lx}(t)\ &=& \ w(t)+Ey(t) \nonumber \\
\hat z(t)\ &=& \ \begin{pmatrix}I_2&0\end{pmatrix}\bigl(w(t)+Ey(t)\bigr) \nonumber
\end{IEEEeqnarray}
estimates $\mathcal Lx(t)$, and hence $z(t)=Lx(t)$, when, with
$P=\mathcal L-EC$,
$$
PA-NP-JC=0,\qquad H=PB,
$$
as in \cite[eqs.~(3)--(4)]{darouach2000}. Equivalently,
$\mathcal LA=N\mathcal L+E\,CA+(J-NE)C$. Here
$\rank\begin{pmatrix}\mathcal L\\CA\\C\end{pmatrix}=7$, so the solutions
$(N,E,J-NE)$ form a family $(N_0,E_0,K_0)+f\mu$, $f\in\mathbb R^4$, where
$\mu=(\mu_{\mathcal L},\mu_{CA},\mu_C)$ spans the left nullspace of this
matrix; in particular $N=N_0+f\mu_{\mathcal L}$, and condition~(S)
guarantees that $(N_0,\mu_{\mathcal L})$ is observable, so the eigenvalues
of $N$ can be assigned freely and determine $f$ uniquely. Choosing
$$
\operatorname{eig}(N)=\{-1,-2,-3,-4\}
$$
gives (to four decimal places)
\begin{IEEEeqnarray}{rcl}
	N&\ = &\
	\begin{pmatrix}
		-2.6648&-5.9466&-9.9348&-0.7508\\
		0.2934&-5.9060&-4.0842&-0.4476\\
		0.5902&0.0374&0.4758&0.0934\\
		0.1141&0.4217&1.2555&-1.9050
	\end{pmatrix} \vspace*{2mm}
	\nonumber\\ 
	J&\ = &\
	\begin{pmatrix}
		-33.9666&-23.0629\\
		-18.4909&-12.0716\\
		3.7491&2.5547\\
		3.2164&1.9338
	\end{pmatrix} \vspace*{2mm}
\nonumber\\
	H &\ = &\
	\begin{pmatrix}
		35.9517\\
		28.4530\\
		-0.4535\\
		27.3886
	\end{pmatrix} \vspace*{2mm}
	\nonumber\\
	E&\ = &\
	\begin{pmatrix}
		5.2693&5.8025\\
		3.8497&3.6022\\
		-0.0380&-0.6317\\
		-0.9703&-0.8139
	\end{pmatrix}.
	\nonumber
\end{IEEEeqnarray}
These are the unique parameters for this $\mathcal L$ with these
eigenvalues, so they can be reproduced exactly from $\mathcal L$. They were
computed in exact rational arithmetic, for which $PA-NP-JC=0$ and $H=PB$
hold exactly and the characteristic polynomial of $N$ is
$(s+1)(s+2)(s+3)(s+4)$; the four-decimal values shown satisfy these
equations to the accuracy of the rounding. This exhibits a complete order-$4$
functional observer at the true minimum order $\nu_{\min}=4$, realized via
the Sylvester-equation construction applied to the augmented functional
$\mathcal L$ rather than to $L$ itself.

\subsection{Interpretation}

Combining the relevant dimensions gives
$$
r=2<q_0=3<\nu_{\min}=4<n_0-p=6.
$$
Here $\nu_{\min}$ is strictly interior to the admissible range: the orders
$q=5$ and $q=6$ remain feasible above the minimum. The example therefore
separates all four quantities $r$, $q_0$, $\nu_{\min}$, and $n_0-p$
simultaneously.

The mechanism is explicit. Condition~(A) confines every order-$q_0$
nullspace to the pencil of three-dimensional subspaces containing
$\mathcal Z_4=\Span\{x_0,Ax_0\}$, and on this pencil the matrix pencil of
Theorem~\ref{thm:S-pencil} always loses rank at the roots of
$g(\lambda,t)$. At order $q_0+1$ the nullspace $\mathcal Z_4$ itself, the
part common to the whole order-$q_0$ family, satisfies both conditions.
This confirms that $q_0$ is an algebraic lower bound, not a guarantee that a
functional observer exists at that order.

For the present example, we also apply the functional-observer design
algorithm of Darouach and Fernando~\cite{darouachfernando2025}, which is
based on sufficient solvability conditions. The order-$r$ existence
condition fails at $q=2$. At $q=3,4,$ and $5$, the corresponding
solvability conditions impose nontrivial restrictions on the
pole-placement matrix $N$ and hence fail for a generic choice of $N$.
At $q=6=n-p$, by contrast, $\mathcal O(C,A,4)$ has full column rank, so
the solvability condition is satisfied for arbitrary $N$. Thus the
algorithm of \cite{darouachfernando2025} generically reaches the
reduced-order Luenberger endpoint $q=6$, whereas the present
order-by-order characterization establishes that the true minimum
functional-observer order is $\nu_{\min}=4$.

\section{Discussion and Conclusion} \label{sec:discussion-conclusion}
\subsection{Discussion} \label{sec:discussion}

The characterization separates three questions that are easily conflated in minimum-order observer design.

First, the FO indices determine the smallest order at which the algebraic condition can hold. Second, functional observability implies that algebraic feasibility then persists at every higher order through $n_0-p$. Third, spectral feasibility is a separate fixed-order question. The nullspace pencil \eqref{eq:S-Z-rank} provides the same spectral test at every one of these predetermined orders.

This distinction also avoids a replacement-versus-augmentation difficulty: a higher-order realization $\mathcal L$ may either \emph{augment} a smaller one, in the sense of containing the rows of a particular minimum-order realization as a submatrix (modulo $C$), or it may be a \emph{replacement} realization whose rows bear no such relationship to any specific smaller solution, yet still satisfy condition~(A) at the higher order in their own right. Consequently, parameterizing only one-row extensions of the complete minimum-(A) family would not in general describe all higher-order realizations. The matrix representation $Z=K_LW$, together with the rank conditions
\eqref{eq:chart-testa}--\eqref{eq:chart-testb}, parameterizes the order-$q$ problem
directly and therefore includes replacement realizations automatically.

The characterization does not assume or use any interval property of
$\Sigma_{FO}(L)$: unlike the algebraic spectrum $\Sigma_A(L)$, whose interval
structure is established in Theorem~\ref{thm:A-spectrum}, functional-observer
existence is determined independently at each admissible order. Whether
$\Sigma_{FO}(L)$ is always an interval is left open.

For the all-$\lambda$ pole-assignability formulation, additional rank screening can be developed from the common columns represented by $Z$ and $V$. Such screening is useful computationally but is not required for the main equivalence. For the detectability formulation, care is required because rank loss at a strictly stable value of $\lambda$ is permitted.

\subsection{Conclusion}\label{sec:conclusion}

This paper has given an order-by-order characterization of functional
observers between the algebraic lower endpoint and the endpoint determined
by functional observability. The FO indices determine
$q_0=\sum_{j=1}^r\eta_j$, the minimum order for which condition~(A) is
feasible, and the one-row extension theorem shows that condition~(A)
remains feasible at every order $q_0,q_0+1,\ldots,n_0-p$ under functional
observability. At each such order, the fixed-order problem is expressed entirely through
matrix ranks, nullspaces, and column spaces in the original state
coordinates, yielding a complete matrix parameterization of the algebraic
and spectral observer conditions.

Evaluating these conditions over
$q\in\{q_0,\ldots,n_0-p\}$ yields the complete functional-observer
order spectrum $\Sigma_{FO}(L)$ and hence the true minimum order
$$
\nu_{\min}=\min\Sigma_{FO}(L).
$$
As the example in Section~\ref{sec:example} demonstrates, the algebraic
lower bound $q_0$ need not itself be feasible: there, every algebraically
admissible order-$q_0$ realization fails condition~(S), and the minimum is
attained at $q_0+1$, strictly below $n_0-p$. Thus $q_0$ determines where the
search begins, whereas $\nu_{\min}$ identifies the first order at which
both the algebraic and spectral conditions are satisfied.

The upper endpoint $n_0-p$ is represented by a basis of
$\ker(\mathcal O_C)$; when $n_0=n$, this basis has no columns and the
endpoint reduces to the classical reduced-order Luenberger order $n-p$.
When $q_0=r$, the lower endpoint reduces to Darouach's prescribed order.
Thus the Darouach and Luenberger orders arise as specializations of the two
endpoints of the same matrix-based functional-observer framework, with
$\nu_{\min}$ characterized as the smallest order in this range at which
the algebraic and spectral conditions hold jointly.

	\section*{Biography}

\begin{wrapfigure}{l}{1in}
	\includegraphics[width=1in,height=1.25in,clip,keepaspectratio]{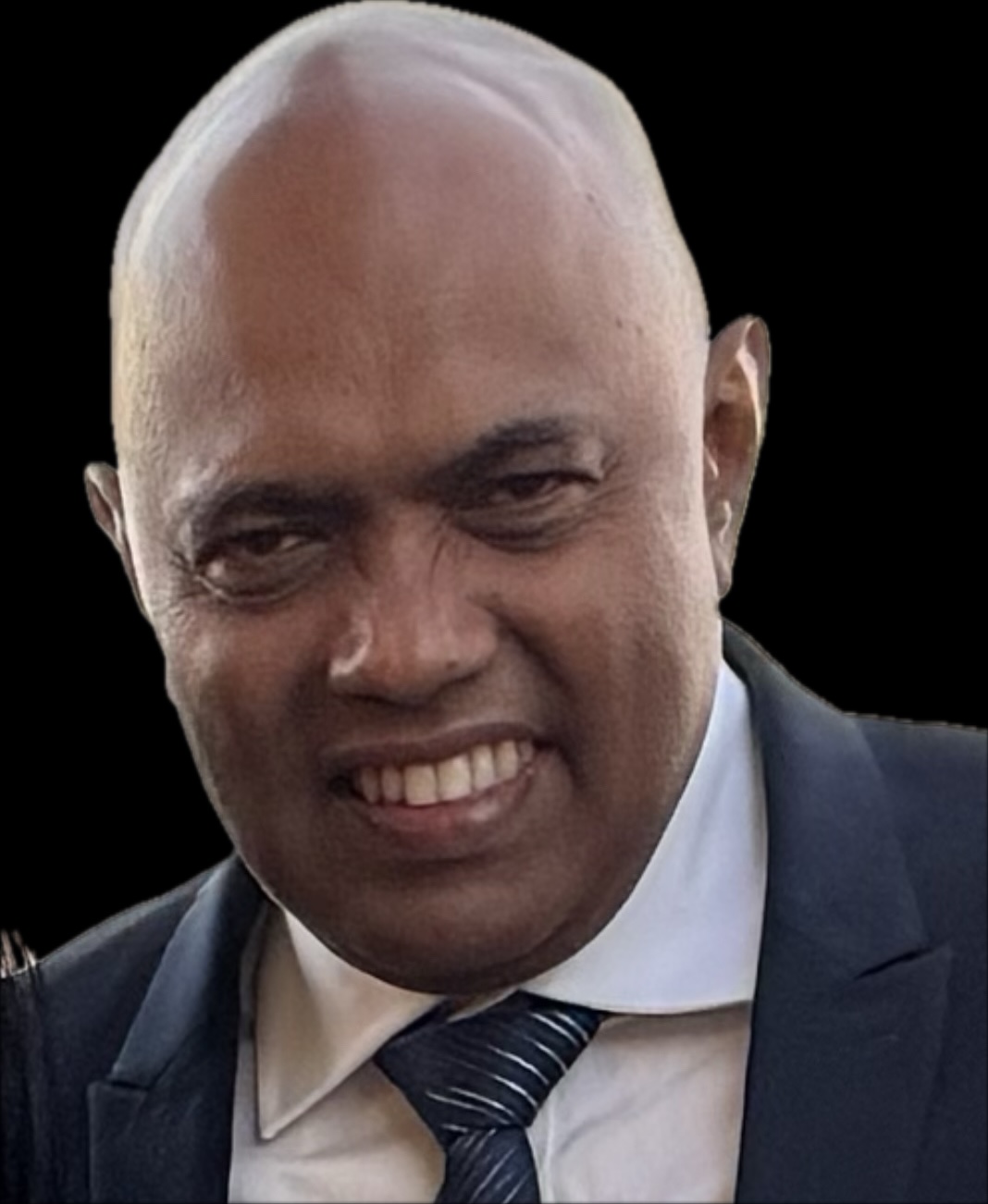}
\end{wrapfigure}
\noindent
\textbf{Tyrone Fernando} received his B.E. (Hons.) and Ph.D. degrees in Electrical Engineering from the University of Melbourne, Victoria, Australia, in 1990 and 1996, respectively. In 1996, he joined the Department of Electrical, Electronic and Computer Engineering, University of Western Australia (UWA), Crawley, WA, Australia, where he currently holds the position of Professor of Electrical Engineering. He previously served as Associate Head and Deputy Head of the department from 2008 to 2010.

Prof. Fernando is currently the Head of the Power and Clean Energy Research Group at UWA. He has provided professional consultancy to Western Power on the integration and management of distributed energy resources in the electric grid. In recognition of his professional contributions, he was named the Outstanding WA IEEE PES/PELS Engineer in 2018. His research interests include theoretical control, observer design, and power system stability and control. He has received multiple teaching awards from UWA in recognition of his contributions to control systems and power systems education.

\end{document}